\documentclass{llncs}

\usepackage{amsfonts}
\usepackage{amsmath}
\usepackage{amssymb}
\usepackage{enumerate}
\usepackage{epsfig}
\usepackage{pseudocode}

\def\Kset{\mathbb{K}}

\def\Rset{\mathbb{R}}
\def\argmin{\mathop{\rm argmin}}

\def\Tset{\Rset_+ \cup \{\infty\}}

\newcommand{\bi}{\begin{itemize}}
\newcommand{\ei}{\end{itemize}}
\newcommand{\be}{\begin{enumerate}}
\newcommand{\ee}{\end{enumerate}}
\newcommand{\bd}{\begin{description}}
\newcommand{\ed}{\end{description}}

\newcommand{\set}[1]{\{#1\}}
\newcommand{\0}{\overline{0}}
\newcommand{\1}{\overline{1}}
\newcommand{\e}{\epsilon}
\newcommand{\+}{\oplus}
\renewcommand{\.}{\otimes}
\newcommand{\K}{(\Kset, \+, \., \0, \1)}

\newcommand{\T}{(\Rset_+ \cup \set{\infty}, \min, +, \infty, 0)}
\newcommand{\R}{(\Rset_+ \cup \set{\infty}, +, \times, 0, 1)}

\newcommand{\level}{\textrm{level}}
\newcommand{\floor}[1]{\lfloor #1 \rfloor}

\newcommand{\Trans}{T = (\Sigma, \Delta, Q, I, F, E, \lambda, \rho)}

\newcommand{\ipsfig}[2]{\scalebox{#1}{\psfig{#2}}}

\newcommand{\ignore}[1]{}

\title{Linear-Space Computation of the Edit-Distance between a String 
  and a Finite Automaton}
\author{Cyril Allauzen \inst{1} \and Mehryar Mohri \inst{2,1}}
\institute{
        Google Research\\
        76 Ninth Avenue, New York, NY 10011, US.\\
        \and
	Courant Institute of Mathematical Sciences\\
	251 Mercer Street, New York, NY 10012, US.\\
}

\begin{document}
\maketitle

\begin{abstract}
  The problem of computing the edit-distance between a string and a
  finite automaton arises in a variety of applications in
  computational biology, text processing, and speech recognition. This
  paper presents linear-space algorithms for computing the
  edit-distance between a string and an arbitrary weighted automaton
  over the tropical semiring, or an unambiguous weighted automaton
  over an arbitrary semiring. It also gives an efficient linear-space
  algorithm for finding an optimal alignment of a string and such a
  weighted automaton.
\end{abstract}

\section{Introduction}

The problem of computing the edit-distance between a string and a
finite automaton arises in a variety of applications in computational
biology, text processing, and speech recognition
\cite{durbin,gusfield,navarro-raffinot,pevzner,edit}. This may be to
compute the edit-distance between a protein sequence and a family of
protein sequences compactly represented by a finite automaton
\cite{durbin,gusfield,pevzner}, or to compute the error rate of a word
lattice output by a speech recognition with respect to a reference
transcription \cite{edit}. A word lattice is a weighted automaton,
thus this further motivates the need for computing the edit-distance
between a string and a weighted automaton. In all these cases, an
optimal alignment is also typically sought. In computational biology,
this may be to infer the function and various properties of the
original protein sequence from the one it is best aligned with. In
speech recognition, this determines the best transcription hypothesis
contained in the lattice.

This paper presents linear-space algorithms for computing the
edit-distance between a string and an arbitrary weighted automaton over
the tropical semiring, or an unambiguous weighted automaton over an
arbitrary semiring. It also gives an efficient linear-space algorithm
for finding an optimal alignment of a string and such a weighted
automaton. Our linear-space algorithms are obtained by using the same
generic shortest-distance algorithm but by carefully defining
different queue disciplines. More precisely, our meta-queue
disciplines are derived in the same way from an underling queue
discipline defined over states with the same level.

The connection between the edit-distance and the shortest distance in
a directed graph was made very early on (see
\cite{gusfield,crochemore-hancart-lecroq,crochemore-rytter,crochemore-rytter02}
for a survey of string algorithms). This paper revisits some of these
algorithms and shows that they are all special instances of the same
generic shortest-distance algorithm using different queue
disciplines. We also show that the linear-space algorithms all
correspond to using the same meta-queue discipline using different
underlying queues. Our approach thus provides a better understanding
of these classical algorithms and makes it possible to easily
generalize them, in particular to weighted automata.

The first algorithm to compute the edit-distance between a string $x$
and a finite automaton $A$ as well as their alignment was due to
Wagner~\cite{wagner} (see also \cite{wagner-seiferas}). Its time
complexity was in $O(|x| |A|_Q^2)$ and its space complexity in
$O(|A|_Q^2 |\Sigma| + |x| |A|_Q)$, where $\Sigma$ denotes the alphabet
and $|A|_Q$ the number of states of $A$. Sankoff and
Kruskal~\cite{sankoff} pointed out that the time and space complexity
$O(|x| |A|)$ can be achieved when the automaton $A$ is acyclic. Myers
and Miller~\cite{myers-miller89} significantly improved on previous
results. They showed that when $A$ is acyclic or when it is a
\emph{Thompson automaton}, that is an automaton obtained from a
regular expression using Thompson's construction \cite{thompson}, the
edit-distance between $x$ and $A$ can be computed in $O(|x| |A|)$ time
and $O(|x| + |A|)$ space. They also showed, using a technique due to
Hirschberg \cite{hirschberg}, that the optimal alignment between $x$
and $A$ can be obtained in $O(|x| + |A|)$ space, and in $O(|x| |A|)$
time if $A$ is acyclic, and in $O(|x| |A| \log |x|)$ time when $A$ is
a Thompson automaton.

The remainder of the paper is organized as
follows. Section~\ref{sec:preliminaries} introduces the definition of
semirings, and weighted automata and transducers. In
Section~\ref{sec:edit-distance}, we give a formal definition of the
edit-distance between a string and a finite automaton, or a weighted
automaton. Section~\ref{sec:algo} presents our linear-space
algorithms, including the proof of their space and time complexity and
a discussion of an improvement of the time complexity for automata
with some favorable graph structure property.

\section{Preliminaries}
\label{sec:preliminaries}

This section gives the standard definition and specifies the notation
used for weighted transducers and automata which we use in our
computation of the edit-distance.

\emph{Finite-state transducers} are finite automata \cite{perrin} in
which each transition is augmented with an output label in addition to
the familiar input label \cite{berstel,eilenberg}. Output labels are
concatenated along a path to form an output sequence and similarly
input labels define an input sequence. \emph{Weighted transducers} are
finite-state transducers in which each transition carries some weight
in addition to the input and output labels
\cite{soittola,kuich}. Similarly, \emph{weighted automata} are finite
automata in which each transition carries some weight in addition to
the input label. A path from an initial state to a final state is
called an \emph{accepting path}. A weighted transducer or weighted
automaton is said to be \emph{unambiguous} if it admits no two
accepting paths with the same input sequence.

The weights are elements of a semiring $\K$, that is a ring that may
lack negation \cite{kuich}. Some familiar semirings are the tropical
semiring $\T$ and the probability semiring $\R$, where $\Rset_+$
denotes the set of non-negative real numbers. In the following, we
will only consider weighted automata and transducers over the tropical
semiring. However, all the results of section~\ref{sec:algo:edit} hold
for an unambiguous weighted automaton $A$ over an arbitrary semiring.

The following gives a formal definition of weighted transducers.

\begin{definition}
  A {\em weighted finite-state transducer} $T$ over the tropical
  semiring $\T$ is an 8-tuple $\Trans$ where $\Sigma$ is the finite
  input alphabet of the transducer, $\Delta$ its finite output
  alphabet, $Q$ is a finite set of states, $I \subseteq Q$ the set of
  initial states, $F \subseteq Q$ the set of final states, $E
  \subseteq Q \times (\Sigma \cup \set{\epsilon}) \times (\Delta \cup
  \set{\epsilon}) \times (\Tset) \times Q$ a finite set of
  transitions, $\lambda: I \rightarrow \Tset$ the initial weight
  function, and $\rho: F \rightarrow \Tset$ the final weight function
  mapping $F$ to $\Tset$.

\end{definition}
We define the {\em size} of $T$ as $|T| = |T|_Q + |T|_E$ where $|T|_Q
= |Q|$ is the number of states and $|T|_E = |E|$ the number of
transitions of $T$.

The weight of a path $\pi$ in $T$ is obtained by summing the weights
of its constituent transitions and is denoted by $w[\pi]$. The weight
of a pair of input and output strings $(x, y)$ is obtained by taking
the minimum of the weights of the paths labeled with $(x, y)$ from an
initial state to a final state.

For a path $\pi$, we denote by $p[\pi]$ its origin state and by
$n[\pi]$ its destination state. We also denote by $P(I, x, y, F)$ the
set of paths from the initial states $I$ to the final states $F$
labeled with input string $x$ and output string $y$.  The weight $T(x,
y)$ associated by $T$ to a pair of strings $(x, y)$ is defined by:
\begin{equation}
T(x, y) = \min_{\pi \in P(I, x, y, F)} \lambda(p[\pi])
+ w[\pi] + \rho(n[\pi]).
\end{equation}
Figure~\ref{fig:examples}(a) shows an example of weighted transducer
over the tropical semiring.

\begin{figure*}[t]
\begin{center}
\begin{tabular}{c@{\hspace{2cm}}c}
\ipsfig{.4}{figure=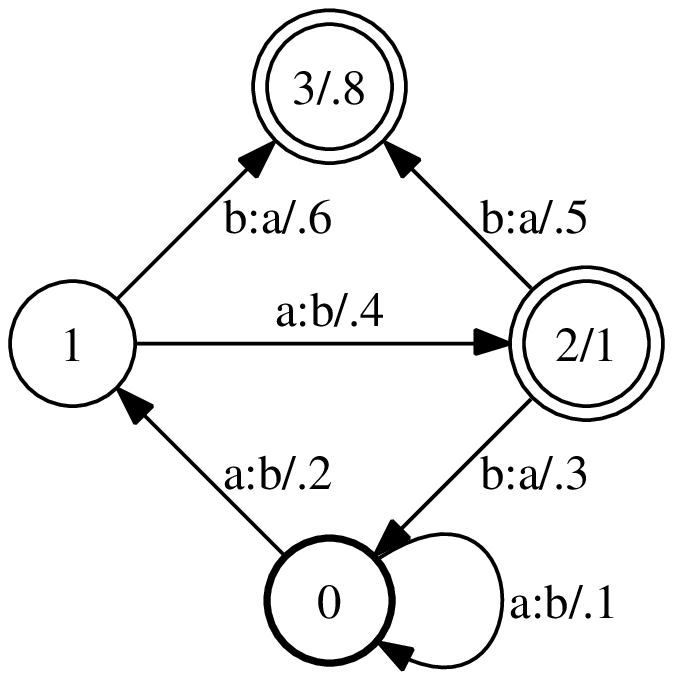} & \ipsfig{.4}{figure=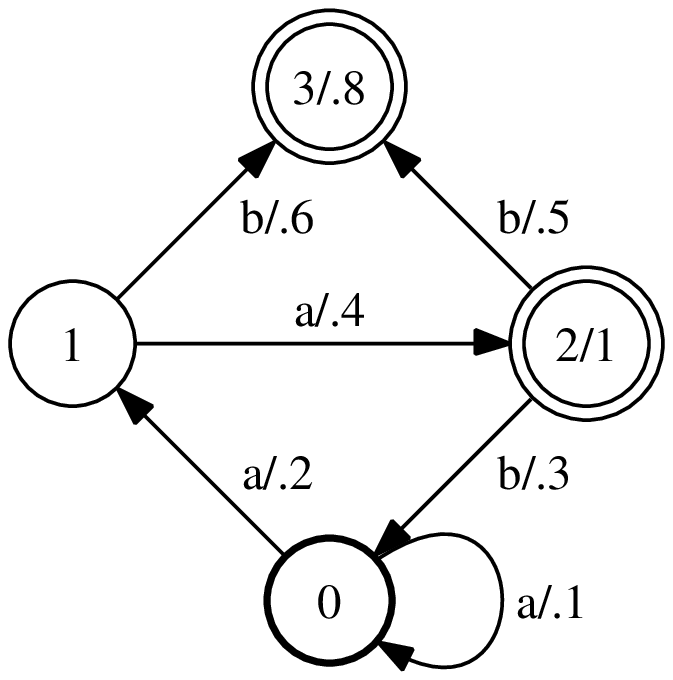}\\
(a) & (b)
\end{tabular}
\end{center}
\caption[]{(a) Example of a weighted transducer $T$. (b) Example of a
weighted automaton $A$. $T(aab, bba) = A(aab) = \min(.1 +
.2 + .6 + .8, .2 + .4 + .5 + .8)$. A bold
circle indicates an initial state and a double-circle a final state.
The final weight $\rho(q)$ of a final state $q$ is indicated after the
slash symbol representing $q$. }
\label{fig:examples}
\end{figure*}
\emph{Weighted automata} can be defined as
weighted transducers $A$ with identical input and output labels, for
any transition. Thus, only pairs of the form $(x, x)$ can have a
non-zero weight by $A$, which is why the weight associated by $A$ to
$(x, x)$ is abusively denoted by $A(x)$ and identified with the
\emph{weight associated by $A$ to $x$}.  Similarly, in the graph
representation of weighted automata, the output (or input) label is
omitted. Figure~\ref{fig:examples}(b) shows an example.

\section{Edit-distance}
\label{sec:edit-distance}

We first give the definition of the edit-distance between a string and
a finite automaton.

Let $\Sigma$ be a finite alphabet, and let $\Omega$ be defined by
$\Omega = (\Sigma \cup \set{\epsilon}) \times (\Sigma \cup
\set{\epsilon}) - \set{(\epsilon, \epsilon)}$. An element of $\Omega$
can be seen as a symbol edit operation: $(a,\epsilon)$ is a deletion,
$(\epsilon, a)$ an insertion, and $(a, b)$ with $a \neq b$ a
substitution.  We will denote by $h$ the natural morphism between
$\Omega^*$ and $\Sigma^* \times \Sigma^*$ defined by $h((a_1, b_1)
\cdots (a_n, b_n)) = (a_1 \cdots a_n, b_1 \cdots b_n)$. An {\em
  alignment} $\omega$ between two strings $x$ and $y$ is an element of
$\Omega^*$ such that $h(\omega) = (x,y)$.

Let $c : \Omega \rightarrow \Rset_+$ be a function associating a
non-negative cost to each edit operation. The cost of an alignment
$\omega = \omega_1 \cdots \omega_n$ is defined as $c(\omega) =
\sum_{i=1}^n c(\omega_i)$.

\begin{definition}
  The {\em edit-distance} $d(x, y)$ of two strings $x$ and $y$ is the
  minimal cost of a sequence of symbols insertions, deletions or
  substitutions transforming one string into the other:
\begin{equation}
d(x, y) = \min_{h(\omega) = (x, y)} c(\omega).
\end{equation}
  When $c$ is the function defined by $c(a, a) = 0$ and $c(a,
  \epsilon) = c(\epsilon, a) = c(a,b) = 1$ for all $a$, $b$ in
  $\Sigma$ such that $a \not= b$, the edit-distance is also known as
  the Levenshtein distance. The {\em edit-distance $d(x, A)$ between a
  string $x$ and a finite automaton $A$} can then be defined as
\begin{equation}
d(x, A) = \min_{y \in L(A)} d(x, y),
\end{equation}
  where $L(A)$ denotes the regular language accepted by $A$. The {\em
  edit-distance $d(x,A)$ between a string $x$ and a weighted automaton
  $A$ over the tropical semiring} is defined as:
\begin{equation}
d(x,A) = \min_{y \in \Sigma^*} \bigl(A(y) + d(x, y)\bigr).
\end{equation}
\end{definition}

\section{Algorithms}
\label{sec:algo}

In this section, we present linear-space algorithms both for computing
the edit-distance $d(x, A)$ between an arbitrary string $x$ and an
automaton $A$, and an optimal alignment between $x$ and $A$, that is
an alignment $\omega$ such that $c(\omega) = d(x, A)$.

We first briefly describe two general algorithms that we will use as
subroutines.

\subsection{General algorithms}

\subsubsection{Composition.}
\label{sec:composition}

The {\em composition} of two weighted transducers $T_1$ and $T_2$ over
the tropical semiring with matching input and output alphabets
$\Sigma$, is a weighted transducer denoted by $T_1 \circ
T_2$ defined by:
\begin{equation}
(T_1 \circ T_2)(x, y) = \min_{z \in \Sigma^*}\ T_1(x,z) + T_2(z, y).
\end{equation}
$T_1 \circ T_2$ can be computed from $T_1$ and $T_2$ using the
composition algorithm for weighted transducers
\cite{pereira-riley,ecai}. States in the composition $T_1 \circ T_2$
are identified with pairs of a state of $T_1$ and a state of $T_2$. In
the absence of transitions with $\e$ inputs or outputs, the
transitions of $T_1 \circ T_2$ are obtained as a result of the
following matching operation applied to the transitions of $T_1$ and
$T_2$:
\begin{equation}
(q_1, a, b, w_1, q_1') \mbox{ and } (q_2, b, c, w_2, q'_2)
\to ((q_1,q_2), a, c, w_1 + w_2, (q'_1,q'_2)).
\end{equation}
A state $(q_1,q_2)$ of $T_1 \circ T_2$ is initial (resp. final) iff
$q_1$ and $q_2$ are initial (resp. final) and, when it is final, its
initial (resp.final) weight is the sum of the initial (resp. final)
weights of $q_1$ and $q_2$. In the worst case, all transitions of
$T_1$ leaving a state $q_1$ match all those of $T_2$ leaving state
$q_2$, thus the space and time complexity of composition is quadratic,
that is $O(|T_1||T_2|)$.

\subsubsection{Shortest distance.}
\label{sec:sd}

Let $A$ be a weighted automaton over the tropical semiring. The {\em
shortest distance} from $p$ to $q$ is defined as
\begin{equation}
d[p, q] = \min_{\pi \in P(p, q)} w[\pi].
\end{equation}
It can be computed using the generic single-source shortest-distance
algorithm of \cite{shortest-distance}, a generalization of the
classical shortest-distance algorithms. This generic shortest-distance
algorithm works with an arbitrary \emph{queue discipline}, that is the
order according to which elements are extracted from a queue.  We
shall make use of this key property in our algorithms. The pseudocode
of a simplified version of the generic algorithm for the tropical
semiring is given in Figure~\ref{fig:alg:shortest}.

\begin{figure}[ht]
{\small 
$\mbox{\sc Shortest-Distance}(A,s)$
\bc
\li \FOR \EACH $p \in Q$ \DO\\
\li \> $d[p] \EQ \infty$\\
\li $d[s] \EQ 0$\\
\li $S \EQ \set{s}$\\
\li \WHILE $S \not= \emptyset$ \DO\\
\li \> $q \EQ \HEAD(S)$ \\
\li \> $\DEQUEUE(S)$ \\
\li \> \FOR \EACH $e \in E[q]$ \DO \\
\li \> \> \IF $(d[s] + w[e] < d[n[e]])$ \THEN\\
\li \> \> \> $d[n[e]] \EQ d[s] + w[e]$\\
\li \> \> \> \IF $(n[e] \not\in S)$   \THEN\\
\li \> \> \> \> $\ENQUEUE(S,n[e])$\\
\ec
}
\caption{Pseudocode of the generic shortest-distance algorithm.}
\label{fig:alg:shortest}
\end{figure}

The complexity of the algorithm depends on the queue discipline
selected for $S$. Its general expression is
\begin{equation}
\label{eq:gen_exp}
  O(|Q|+ \mathsf{C(A)} \max_{q \in Q} \mathsf{N(q)} |E| + (\mathsf{C(I)} + 
  \mathsf{C(X)}) \sum_{q\in Q} \mathsf{N(q)}),
\end{equation}
where $\mathsf{N(q)}$ denotes the number of times state $q$ is
extracted from queue $S$, $\mathsf{C(X)}$ the cost of extracting a
state from $S$, $\mathsf{C(I)}$ the cost of inserting a state in $S$,
and $\mathsf{C(A)}$ the cost of an assignment.

With a shortest-first queue discipline implemented using a heap, the
algorithm coincides with Dijkstra's algorithm \cite{dijkstra} and its
complexity is $O((|E|+|Q|)\log |Q|)$. For an acyclic automaton and
with the topological order queue discipline, the algorithm coincides
with the standard linear-time ($O(|Q|+|E|)$) shortest-distance
algorithm \cite{rivest}.

\subsection{Edit-distance algorithms}
\label{sec:algo:edit}

The edit cost function $c$ can be naturally represented by a one-state
weighted transducer over the tropical semiring $T_c = (\Sigma, \Sigma,
\set{0}, \set{0}, \set{0}, E_c, \1, \1)$, or $T$ in the absence of
ambiguity, with each transition corresponding to an edit operation:
$E_c = \set{ (0, a, b, c(a, b), 0) | (a, b) \in \Omega}$.

\begin{lemma}
\label{lem:edit}
  Let $A$ be a weighted automaton over the tropical semiring and let
  $X$ be the finite automaton representing a string $x$. Then, the
  edit-distance between $x$ and $A$ is the shortest-distance from the
  initial state to a final state in the weighted transducer $U = X
  \circ T \circ A$.
\end{lemma}
\begin{proof}
  Each transition $e$ in $T$ corresponds to an edit operation $(i[e],
  o[e]) \in \Omega$, and each path $\pi$ corresponds to an alignment
  $\omega$ between $i[\pi]$ and $o[\pi]$. The cost of that alignment
  is, by definition of $T$, $c(\omega) = w[\pi]$. Thus, $T$ defines
  the function:
\begin{equation}
  T(u,v) = \min_{\omega \in \Omega^*} \set{ c(\omega)\colon h(\omega) = (u,v)}
  = d(u,v),
\end{equation}
  for any strings $u$, $v$ in $\Sigma^*$. Since $A$ is an automaton and
  $x$ is the only string accepted by $X$, it follows from the definition
  of composition that $U(x,y) = T(x,y) + A(y) = d(x,y) + A(y)$.
  The shortest-distance from the initial state to a final state in $U$ is
  then:
\begin{align}
\min_{\pi \in P_U(I, F)} w[\pi] &=  \min_{y \in \Sigma^*} 
  \min_{\pi \in P_U(I, x, y, F)}  w[\pi] 
 = \min_{y \in \Sigma^*} U(x,y) \\ 
&= \min_{y \in \Sigma^*} \bigl( d(x,y) + A(y) \bigr) = d(x, A),
\end{align}
  that is the edit-distance between $x$ and $A$.\qed
\end{proof}

\begin{figure*}[t]
\begin{center}
\begin{tabular}{c@{\hspace{2cm}}c@{\hspace{2cm}}c}
\ipsfig{.45}{figure=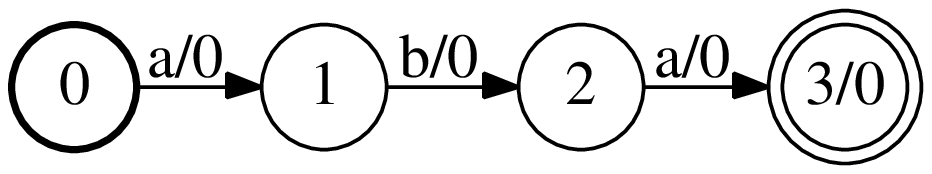} & \ipsfig{.45}{figure=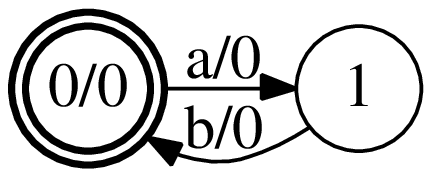}\\
(a) & (b)   \\
\ipsfig{.45}{figure=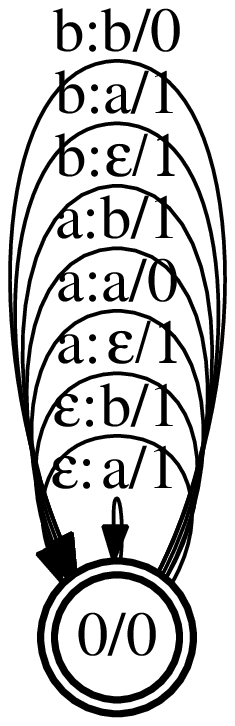} & \hspace{-5cm}\ipsfig{.5}{figure=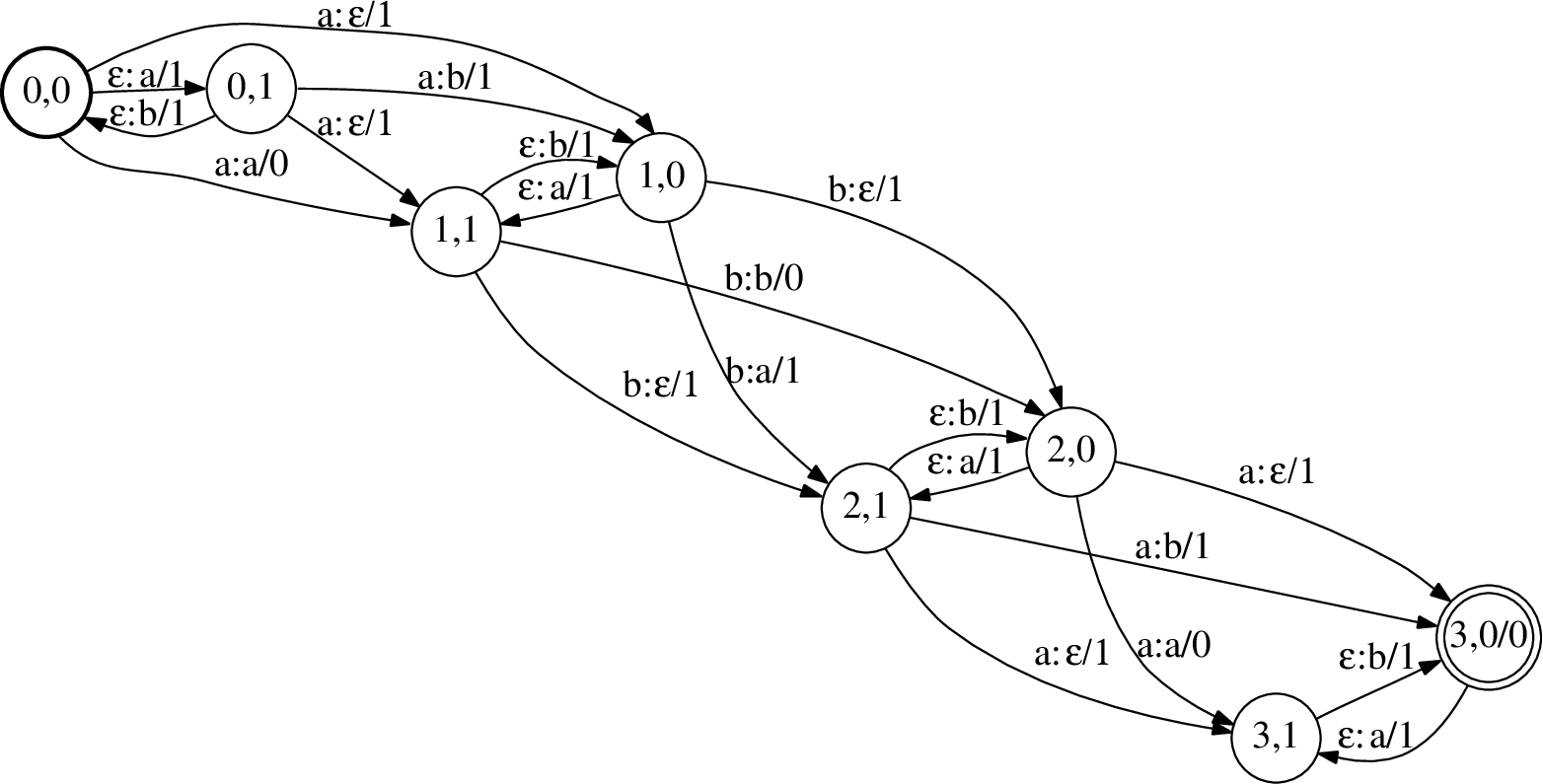}\\
(c) & (d)
\end{tabular}
\end{center}
\caption[]{ (a) Finite automaton $X$ representing the string $x =
aba$. (b) Finite automaton $A$. (c) Edit transducer $T$ over the
alphabet $\set{a,b}$ where the cost of any insertion, deletion and
substitution is 1.  (d) Weighted transducer $U = X \circ T \circ A$.}
\label{fig:edit}
\end{figure*}

Figure~\ref{fig:edit} shows an example illustrating Lemma~\ref{lem:edit}.
Using the lateral strategy of the 3-way composition algorithm of
\cite{3way} or an \emph{ad hoc} algorithm exploiting the structure of
$T$, $U = X \circ T \circ A$ can be computed in $O(|x| |A|)$ time. The
shortest-distance algorithm presented in Section~\ref{sec:sd} can then
be used to compute the shortest distance from an initial state of $U$
to a final state and thus the edit distance of $x$ and $A$. Let us
point out that different queue disciplines in the computation of that
shortest distance lead to different algorithms and complexities.  In
the next section, we shall give a queue discipline enabling us to
achieve a linear-space complexity.

\subsection{Edit-distance computation in linear space}

Using the shortest-distance algorithm described in
Section~\ref{sec:sd} leads to an algorithm with space complexity
linear in the size of $U$, i.e. in $O(|x| |A|)$. However, taking
advantage of the topology of $U$, it is possible to design a queue
discipline that leads to a linear space complexity $O(|x| + |A|)$.

We assume that the finite automaton $X$ representing the string $x$ is
topologically sorted. A state $q$ in the composition $U = X \circ T
\circ A$ can be identified with a triplet $(i, 0, j)$ where $i$ is a
state of $X$, 0 the unique state of $T$, and $j$ a state of $A$. Since
$T$ has a unique state, we further simplify the notation by
identifying each state $q$ with a pair $(i, j)$. For a state $q = (i,
j)$ of $U$, we will refer to $i$ by the \emph{level of $q$}. A key
property of the levels is that there is a transition in $U$ from $q$
to $q'$ iff $\level(q') = \level(q)$ or $\level(q') = \level(q) + 1$.
Indeed, a transition from $(i,j)$ to $(i',j')$ in $U$ corresponds to
taking a transition in $X$ (in that case $i' = i + 1$ since $X$ is
topologically sorted) or staying at the same state in $X$ and
taking an input-$\epsilon$ transition in $T$ (in that case $i' = i$).

From any queue discipline $\prec$ on the states of $U$, we can derive
a new queue discipline $\prec_l$ over $U$ defined for all $q, q'$ in
$U$ as follows:
\begin{equation}
\label{eq:queue}
q \prec_l q' \text{ iff $\bigl(\level(q) < \level(q')\bigr)$ or
$\bigl(\level(q) = \level(q')$ and $q \prec q'\bigr)$}.
\end{equation}

\begin{proposition}
\label{prop:algo}
Let $\prec$ be a queue discipline that requires at most $O(|V|)$ space
to maintain a queue over any set of states $V$. Then, the
edit-distance between $x$ and $A$ can be computed in linear space,
$O(|x|+|A|)$, using the queue discipline $\prec_l$.
\end{proposition}
\begin{proof}
  The benefit of the queue discipline $\prec_l$ is that when computing
  the shortest distance to $q = (i,j)$ in $U$, only the shortest
  distances to the states in $U$ of level $i$ and $i-1$ need to be
  stored in memory. The shortest distances to the states of level
  strictly less than $i - 1$ can be safely discarded. Thus, the space
  required to store the shortest distances is in $O(|A|_Q)$.
 
  Similarly, there is no need to store in memory the full transducer
  $U$. Instead, we can keep in memory the last two levels active in
  the shortest-distance algorithm. This is possible because the
  computation of the outgoing transitions of a state with level $i$
  only requires knowledge about the states with level $i$ and $i +
  1$. Therefore, the space used to store the active part of $U$ is in
  $O(|A|_E + |A|_Q) = O(|A|)$.  Thus, it follows that the space
  required to compute the edit-distance of $x$ and $A$ is linear, that
  is in $O(|x| + |A|)$.\qed
\end{proof}

The time complexity of the algorithm depends on the underlying queue
discipline $\prec$.  A natural choice is for $\prec$ is the
shortest-first queue discipline, that is the queue discipline used in
Dijkstra's algorithm. This yields the following corollary.
\begin{corollary}
\label{cor:cyclic}
The edit-distance between a string $x$ and an automaton $A$ can be
computed in time $O(|x| |A| \log |A|_Q)$ and space $O(|x| + |A|)$
using the queue discipline $\prec_l$.
\end{corollary}
\begin{proof}
  A shortest-first queue is maintained for each level and contains at
  most $|A|_Q$ states. The cost for the global queue of an insertion,
  $\mathsf{C(I)}$, or an assignment, $\mathsf{C(A)}$, is in $O(\log
  |A|_Q)$ since it corresponds to inserting in or updating one of the
  underlying level queues. Since $\mathsf{N(q)} = 1$, the general
  expression of the complexity (\ref{eq:gen_exp}) leads to an overall
  time complexity of $O(|x| |A| \log |A|_Q)$ for the shortest-distance
  algorithm.\qed
\end{proof}

When the automaton $A$ is acyclic, the time complexity can be further
improved by using for $\prec$ the topological order queue discipline.
\begin{corollary}
\label{cor:acyclic}
If the automaton $A$ is acyclic, the edit-distance between $x$ and $A$
can be computed in time $O(|x| |A|)$ and space $O(|x| + |A|)$ using
the queue discipline $\prec_l$ with the topological order queue
discipline for $\prec$.
\end{corollary}
\begin{proof}
  Computing the topological order for $U$ would require $O(|U|)$
  space.  Instead, we use the topological order on $A$, which can be
  computed in $O(|A|)$, to define the underlying queue discipline.
  The order inferred by (\ref{eq:queue}) is then a topological order
  on $U$.  \qed
\end{proof}

Myers and Miller \cite{myers-miller89} showed that when $A$ is a
Thompson automaton, the time complexity can be reduced to $O(|x| |A|)$
even when $A$ is not acyclic.  This is possible because of the
following observation: in a weighted automaton over the tropical
semiring, there exists always a shortest path that is \emph{simple},
that is with no cycle, since cycle weights cannot decrease 
path weight.

In general, it is not clear how to take advantage of this
observation. However, a Thompson automaton has additionally the
following structural property: a \emph{loop-connectedness} of one.
The \emph{loop-connectedness} of $A$ is $k$ if in any depth-first
search of $A$, a simple path goes through at most $k$ back
edges. \cite{myers-miller89} showed that this property, combined with
the observation made previously, can be used to improve the time
complexity of the algorithm.  The results of \cite{myers-miller89} can
be generalized as follows.

\begin{corollary}
  If the loop-connectedness of $A$ is $k$, then the edit-distance
  between $x$ and $A$ can be computed in $O(|x| |A| k)$ time and
  $O(|x| + |A|)$ space.
\end{corollary}
\begin{proof}
  We first use a depth-first search of $A$, identify back edges, and
  mark them as such. We then compute the topological order for $A$,
  ignoring these back edges.  Our underlying queue discipline $\prec$
  is defined such that a state $q = (i,j)$ is ordered first based on
  the number of times it has been enqueued and secondly based on the
  order of $j$ in the topological order ignoring back edges. This
  underlying queue can be implemented in $O(|A|_Q)$ space with
  constant time costs for the insertion, extraction and updating
  operations. The order $\prec_l$ derived from $\prec$ is then not
  topological for a transition $e$ iff $e$ was obtained by matching a
  back edge in $A$ and $\level(p[e]) = \level(n[e])$. When such a
  transition $e$ is visited, $n[e]$ is reinserted in the queue.

  When state $q$ is dequeued for the $l$th time, the value of $d[q]$
  is the weight of the shortest path from the initial state to $q$
  that goes through at most $l - 1$ back edges.  Thus, the inequality
  $\mathsf{N(q)} \le k + 1$ holds for all $q$ and, since the costs for
  managing the queue, $\mathsf{C(I)}$, $\mathsf{C(A)}$, and
  $\mathsf{C(X)}$, are constant, the time complexity of the algorithm
  is in $O(|x| |A| k)$.  \qed
\end{proof}

\subsection{Optimal alignment computation in linear space}

The algorithm presented in the previous section can also be used to
compute an optimal alignment by storing a back pointer at each state
in $U$. However, this can increase the space complexity up to $O(|x|
|A|_Q)$.  The use of back pointers to compute the best alignment can
be avoided by using a technique due to Hirschberg \cite{hirschberg},
also used by \cite{myers-miller88,myers-miller89}.

As pointed out in previous sections, an optimal alignment between $x$
and $A$ corresponds to a shortest path in $U = X \circ T \circ A$.
We will say that a state $q$ in $U$ is a {\em midpoint} of an optimal
alignment between $x$ and $A$ if $q$ belongs to a shortest path in
$U$ and $\level(q) = \floor{|x|/2}$.

\begin{lemma}
Given a pair $(x, A)$, a midpoint of the optimal alignment between $x$
and $A$ can be computed in $O(|x| + |A|)$ space with a time complexity in
$O(|x| |A|)$ if $A$ is acyclic and in $O(|x| |A| \log |A|_Q)$
otherwise.
\end{lemma}
\begin{proof}
  Let us consider $U = X \circ T \circ A$. For a state $q$ in $U$ let
  $d[q]$ denote the shortest distance from the initial state to $q$,
  and by $d^R[q]$ the shortest distance from $q$ to a final state.
  For a given state $q = (i,j)$ in $U$, $d[(i,j)] + d^R[(i,j)]$ is the
  cost of the shortest path going through $(i,j)$.  Thus, for any $i$,
  the edit-distance between $x$ and $A$ is $d(x, A) = \min_{j}
  (d[(i,j)] + d^R[(i,j)])$.

  For a fixed $i_0$, we can compute both $d[(i_0, j)]$ and
  $d^R[(i_0,j)]$ for all $j$ in $O(|x| |A| \log |A|_Q)$ time (or
  $O(|x| |A|$ time if $A$ is acyclic) and in linear space $O(|x|+|A|)$
  using the algorithm from the previous section forward and backward
  and stopping at level $i_0$ in each case. Running the algorithm
  backward (exchanging initial and final states and permuting the
  origin and destination of every transition) can be seen as computing
  the edit-distance between $x^R$ and $A^R$, the {\em mirror images}
  of $x$ and $A$.
 
  Let us now set $i_0 = \floor{|x|/2}$ and $j_0 = \argmin_j (d[(i_0,j)]
  + d^R[(i_0,j)])$. It then follows that $(i_0, j_0)$ is a midpoint of
  the optimal alignment. Hence, for a pair $(x, A)$, the running-time
  complexity of determining the midpoint of the alignment is in $O(|x|
  |A|)$ if $A$ is acyclic and $O(|x| |A| \log |A|_Q)$ otherwise. \qed
\end{proof}

The algorithm proceeds recursively by first determining the midpoint
of the optimal alignment.  At step 0 of the recursion, we first find
the midpoint $(i_0,j_0)$ between $x$ and $A$. Let $x^1$ and $x^2$ be
such that $x = x^1 x^2$ and $|x^1| = i_0$, and let $A^1$ and $A^2$ be
the automaton obtained from $A$ by respectively changing the final
state to $j_0$ in $A^1$ and the initial state to $j_0$ in $A^2$. We
can now recursively find the alignment between $x^1$ and $A^1$ and
between $x^2$ and $A^2$.

\begin{theorem}
  An optimal alignment between a string $x$ and an automaton $A$ can
  be computed in linear space $O(|x| + |A|)$ and in time $O(|x||A|)$
  if $A$ is acyclic, $O(|x||A| \log |x| \log |A|_Q)$ otherwise.
\end{theorem}
\begin{proof}
  We can assume without loss of generality that the length of $x$ is a
  power of 2.  At step $k$ of the recursion, we need to compute the
  midpoints for $2^k$ string-automaton pairs $(x^i_k, A^i_k)_{1 \le i
    \le 2^k}$. Thus, the complexity of step $k$ is in
  $O(\sum_{i=1}^{2^k} |x^i_k| |A^i_k| \log |A^i_k|_Q) = O(
  \frac{|x|}{2^k} \sum_{i=1}^{2^k} |A^i_k| \log |A^i_k|_Q)$ since
  $|x^i_k| = |x|/2^k$ for all $i$.  When $A$ is acyclic, the $\log$
  factor can be avoided and the equality $\sum_{i=1}^{2^k} |A^i_k| =
  O(|A|)$ holds, thus the time complexity of step $k$ is in
  $O(|x||A|/2^k)$. In the general case, each $|A^i_k|$ can be in the
  order of $|A|$, thus the complexity of step $k$ is in $O(|x||A| \log
  |A|_Q)$.

  Since there are at most $\log |x|$ steps in the recursion, this
  leads to an overall time complexity in $O(|x| |A|)$ if $A$ is
  acyclic and $O(|x| |A| \log |A|_Q \log |x|)$ in general.\qed
\end{proof}
When the loop-connectedness of $A$ is $k$, the time complexity can
be improved to $O(k |x| |A| \log |x|)$ in the general case.

\section{Conclusion}

We presented general algorithms for computing in linear space both the
edit-distance between a string and a finite automaton and their
optimal alignment. Our algorithms are conceptually simple and make use
of existing generic algorithms. Our results further provide a better
understanding of previous algorithms for more restricted automata by
relating them to shortest-distance algorithms and general queue
disciplines.

\bibliographystyle{abbrv}
\bibliography{lse}
\end{document}